\documentclass[12pt, a4paper]{article}

\usepackage{amsmath}
\usepackage{amssymb}
\usepackage{amsthm}
\usepackage{graphicx}
\usepackage{enumerate}
\usepackage{mathtools}
\usepackage{authblk}
\usepackage{hyperref}
\theoremstyle{plain}
\newtheorem{stam}{STAM}
\newtheorem{lem}[stam]{Lemma}
\newtheorem{thm}[stam]{Theorem}

\newtheorem{claim}[stam]{Claim}

\newtheorem*{lem*}{Lemma}
\newtheorem*{thm*}{Theorem}
\newtheorem*{prop*}{Proposition}
\newtheorem*{claim*}{Claim}
\newtheorem*{cor*}{Corollary}
\newtheorem*{conj*}{Conjecture}
\newtheorem*{obs*}{Observation}
\theoremstyle{definition}

\newtheorem*{definition*}{Definition}
\newtheorem*{notation*}{Notation}

\theoremstyle{remark}

\newtheorem*{rem*}{\textbf{Remark}}

\newtheorem*{exercise*}{\textbf{Exercise}}
\numberwithin{equation}{section}

\begin{document}
\title{The Lipschitz Constant of Perturbed Anonymous Games}
\author[1]{Ron Peretz}
\author[1]{Amnon Schreiber}
\affil[1]{Bar Ilan University, Israel}
\author[2]{Ernst Schulte-Geers}
\affil[2]{Federal Office for Information Security, Germany}

\maketitle
\begin{abstract}
The worst-case Lipschitz constant of an $n$-player $k$-action $\delta$-perturbed game, $\lambda(n,k,\delta)$, is given an explicit probabilistic description. In the case of $k\geq 3$, $\lambda(n,k,\delta)$ is identified with the passage probability of a certain symmetric random walk on $\mathbb Z$. In the case of $k=2$ and $n$ even, $\lambda(n,2,\delta)$ is identified with the probability that two two i.i.d.\ Binomial random variables are equal. The remaining case, $k=2$ and $n$ odd, is bounded through the adjacent (even) values of $n$. Our characterisation implies a sharp closed form asymptotic estimate of $\lambda(n,k,\delta)$ as $\delta n /k\to\infty$.
\end{abstract}
\section{Introduction}
The Lipschitz constant of a game measures the maximal amount of influence that one player has on the payoff of some other player. Identifying classes of games that admit a small Lipschitz constant is important due to the stability and robustness of their equilibria \cite{kalai04,azrieli-shmaya13}. The Lipschitz constant  is given an explicit description in the class of  perturbed anonymous games (see Theorems~\ref{theorem: k=3},  \ref{theorem: k=2}, and \ref{theorem:main}).

Schmeidler \cite{schmeidler73} taught us that games with a continuum of anonymous players always admit a Nash equilibrium in pure strategies. Since a continuum of players is an idealisation a large finite set of players, it is reasonable to believe that large finite anonymous games should admit an approximation of a pure Nash equilibrium of some sort. Of what sort and how fast this approximation emerges (as the number of players grows)? These questions are given precise answers in Theorem~\ref{theorem: approximate nash}.

Before explaining our notion of approximation let us start with a na\"ive attempt: the notion of a pure approximate  ($\epsilon$-) Nash equilibrium. Perhaps every large enough anonymous game admits an approximate Nash equilibrium. Well, let's see why not. Consider a game in which the players are people who decide whether to go to a party or not. For some reason some of the  people prefer parties with an even number of participants while others prefer an odd number. This game is anonymous, since the players don't care about the identity of the party participants but only about their number. Alas, this game does not admit any pure Nash equilibrium, not even an approximate one, regardless of the number of players. The instability of this game stems from the persistence of its Lipschitz constant. The influence of a single player on another player's payoff remains the same regardless of the number of players.

The notion that does do the trick is that of an approximate Nash equilibrium in \emph{perturbed} pure strategies. A perturbed pure strategy is deviation from a pure strategy to the uniformly mixed strategy  with some (small $\delta>0$) probability. Assume all of the players in our example play perturbed pure strategies. It is now clear that the size of the game matters. When the number of players is small it is likely that none of the players will play randomly and therefore there is no approximate equilibrium. However, as the number of players grows, it becomes more and more likely that at least one of the players will randomise and therefore all players become almost indifferent between going to the party or not; and therefore a pure approximate Nash equilibrium exists (in fact, any perturbed pure strategy profile will constitute an approximate Nash equilibrium). 

The trick of perturbing all players' actions works for anonymous games generally. The rigorous explanation relies on analysis of the Lipschitz constant of the perturbed game. The accumulative effect of many small perturbation is the reduction of the Lipschitz constant the game and, thus, the emergence of a pure approximate Nash equilibrium, which translates to a perturbed pure approximate Nash  equilibrium in the original (unperturbed) game.

Given parameters $n$, $k$, and $\delta$, we give an explicit expression for the worst-case (largest) Lipschitz constant of any $n$-player $k$-action $\delta$-perturbed anonymous game. The expression is given in terms of a symmetric random walk on the integers. For $k\geq 3$, the expression is the tail probability of the first passage time (from 0 to 1). For $k=2$ and $n$ is even, the expression is the probability that a certain random walk lands at $0$ at time $n/2 -1$. When $n$ is odd, we don't have an exact expression, only upper- and lower-bounds that use the adjacent (even) values of $n$.

The Lipschitz constant of perturbed anonymous games has algorithmic applications, as well. Goldberg and Turchetta \cite{goldberg-turchetta17} presented an efficient algorithm for computing approximate Nash equilibrium in n-player 2-action anonymous games. Their algorithm relies on the existence of an approximate equilibrium that uses perturbed pure strategies (each action gets replaced by the uniformly mixed strategy with some small positive probability $\delta$). The existence of such an equilibrium is guaranteed (due to Azrieli and Shmaya~\cite{azrieli-shmaya13}) since perturbed anonymous games admit a small Lipschitz constant. The premise of the method of Goldberg and Turchetta~\cite{goldberg-turchetta17} depends on how tightly one estimates the Lipschitz constant of the perturbed game. Goldberg and Turchetta~\cite{goldberg-turchetta17} obtained an inverse polynomial upper-bound (in $n$, the number of players, assuming 2 actions for each player) which enabled them to prove that their algorithm was polynomial. Cheng et al.~\cite{cheng-et-al17}, improved the upper-bound and extended it to any number of actions, $k$, showing that the Lipschitz constant is $\mathcal{\tilde O}\left(\sqrt{k^9(\delta n)^{-1}}\right)$. We provide an asymptotically sharp approximation for the worst-case Lipschitz constant $\lambda=\lambda(n,k,\delta)$ by identifying it with a passage time of a certain symmetric random walk on $\mathbb Z$. For example, our characterisation implies that $\lambda= \mathcal O\left(\sqrt{k(\delta n)^{-1}}\right)$, as $\delta+k(\delta n)^{-1}\to 0$.

\section{Definitions and results}
\subsection{Lipschitz constant}
An $n$-player $k$-action game is a function $g\colon [k]^n\to[0,1]^n$. Following Azrieli and Shmaya~\cite{azrieli-shmaya13}, the Lipschitz constant of a game is the maximal change in some players payoff when a single opponent
changes his strategy. 

Formally, the Hamming distance between two pure strategy profiles $a,b\in[k]^n$ is defined as
\[
\rho(a,b)=|\{i\in[n]:a_i\neq b_i\}|.
\]
The Lipschitz constant of $g$ is defined as
\[
\lambda(g)=\max |g_{i}(a)-g_i(b)|,
\]
where the maximum is over all $i\in[n]$ and $a,b\in[k]^n$ such that $a_i=b_i$ and $\rho(a,b)=1$.

\subsection{Perturbation}
For $0< \delta < 1$, the $\delta$-perturbation of a strategy $a_i\in[k]$ is the following mixture of $a_i$ and the Uniform distribution $u\sim \mathrm{Uniform}([k])$,
\[
a_i^\delta=(1-\delta)a_i + \delta u.
\]
The $\delta$-perturbation of $g$ is the game $g^\delta\colon[k]^n\to[0,1]^n$ defined by
\[
g^\delta(a_1,\ldots,a_n)=E\left[g(a^\delta_1,\ldots,a^\delta_n)\right].
\]

\subsection{Anonymous games}
A game $g$ is called \emph{anonymous} if, for every $i\in[n]$, $g_i(\cdot)$ is a function of $i$'s own action and the number of other players who take each action $j\in[k]$. Formally, $g$ is anonymous if $g_i(a)=g_i(b)$, for every $i\in[n]$ and every $a,b\in[k]^n$ such that $a_i=b_i$ and $|\{i'\in[n]:a_{i'}=j\}| =|\{i'\in[n]:b_{i'}=j\}|$, for every $j\in[k]$.
\subsection{Symmetric random walk on the integers}
The statement of our first result uses the notion of a symmetric random walk on $\mathbb Z$ with (stationary) rate $r$, which is a sequence of random variables, $S^r_0,S^r_1,\ldots$, whose law is defined by
\begin{align*}
    &P(S^r_0=0)=1,\\
    &P(S_{n+1}^r-S_n^r=0|S_n^r)=1-r,\\
    &P(S_{n+1}^r-S_n^r=+1|S_n^r)=P(S_{n+1}^r-S_n^r=-1|S_n^r)=\tfrac r 2.
\end{align*}

\subsection{Our results}
Our objective is to characterise the  worst-case Lipschitz constant of anonymous games defined by
\[
\lambda(n,k,\delta)=\max \lambda(g^\delta),
\]
where the maximum is over all $n$-player $k$-action anonymous games.

For games with $k\geq 3$ actions we obtain the following characterisation.
\begin{thm}
\label{theorem: k=3}
For every $n\geq 2$, $k\geq 3$, and $\delta\in(0,1)$,
\[
\lambda(n,k,\delta)=(1-\delta)P(S_{n-2}^{2\delta/k}\in\{0,1\}).
\]
\end{thm}
For games with two actions we have an exact characterisation when the number of players is even and an estimation when it is odd.
\begin{thm}
\label{theorem: k=2}
For every $n\in\mathbb N$, and $\delta\in(0,1)$ let us abbreviate $\lambda_n=\lambda(n,2,\delta)$. Then,
\[
\lambda_{2n}=(1-\delta)P\left(S_{n-1}^{\delta(1-\delta/2)}=0\right),
\]
and
\[
\lambda_{2n+1}\in\left[\lambda_{2n+2},\sqrt{\lambda_{2n}\lambda_{2n+2}}\right].
\]
\end{thm}

We obtain the following asymptotically sharp approximation for the case that $n$ is large relative to $k$ and $\delta^{-1}$.
\begin{thm}\label{theorem:main}
For $k\geq 3$, 
\[
\lim_{\frac{n\delta}{k}\to \infty}(1-\delta)^{-1}\sqrt{\frac{\pi n\delta}{k}}\ \times \ \lambda(n,k,\delta)=1.
\]
For $k=2$,
\[
\lim_{{n\delta}\to \infty}(1-\delta)^{-1}\sqrt{\pi n\delta(1-\delta/2)}\ \times \ \lambda(n,2,\delta)=1.
\]
\end{thm}

The following theorem says that anonymous games with a large number of players admit an approximate Nash equilibrium in perturbed pure strategies. 
\begin{thm}\label{theorem: approximate nash}
Every $n$-player $k$-action game admits an $\epsilon$-Nash equilibrium in $\delta$-perturbed pure strategies, whenever $\epsilon\geq\delta+2k\lambda(n,k,\delta)$.

Furthermore, there exist functions $\epsilon(n,k),\delta(n,k)=\mathcal O(kn^{-\frac 1 3})$, such that every $n$-player $k$-action game admits an $\epsilon(n,k)$-Nash equilibrium in $\delta(n,k)$-perturbed pure strategies.
\end{thm}
\section{Preliminaries}
\subsection{The reflection principle}
A symmetric random walk on $\mathbb Z$ is a sequence of random variables, $S_1,S_2\ldots$, such that the increments $I_i:=S_i-S_{i-1}$ (where $S_0:=0$) satisfy
\begin{itemize}
    \item $I_1,I_2,\ldots\in\{0,1,-1\}$,
    \item $I_1,I_2,\ldots$ are mutually independent,
    \item $E[I_i]=0$, for all $i$.
\end{itemize}
We will use the following property of symmetric random walks.\footnote{The reflection principle has become folklore in the theory of random walks. It is often attributed to the French Mathematician D\'esir\'e Andr\'e, who has used it slightly differently than the way we do here. Lemma~\ref{lemma: reflection principle} is very similar to Lemma 3.3.1 in \cite[p. 76]{feller}.} 
\begin{lem}[Reflection Principle]\label{lemma: reflection principle}
Let $S_1,\ldots,S_n$ be a symmetric random walk on $\mathbb Z$, then \[
P(S_1<1,\ldots,S_n< 1)=P(S_n\in\{0,1\}).
\]
\end{lem}
\begin{proof}
Let $T=\min\{t\in\mathbb N: S_t=1\}$. The event $\{T\leq n\}$ is the complement of the event $\{S_1<1,\ldots,S_n< 1\}$, and
\begin{multline*}
    P(T\leq n)=P(S_n>1,T\leq n)+P(S_n<1,T\leq n)+P(S_n=1,T\leq n)\\
    =2P(S_n>1,T\leq n)+P(S_n=1,T\leq n)=2P(S_n>1)+P(S_n=1)\\
    =P(S_n>1)+P(S_n<-1)+P(S_n=1)\\
    =P(S_n\notin\{-1,0\})=P(S_n\notin\{0,1\}).
\end{multline*}
\end{proof}
\subsection{The Poisson Binomial distribution}
A Standard Poisson Binomial random variable is a finite sum of independent (not necessarily identically distributed) Bernoulli random variables. We define a Poisson Binomial (PB) random variable as the sum of a Standard Poisson Binomial random variable and an integer. Note the if $X$ and $Y$ are PB random variables, so are $X+Y$ and $X-Y$. The distribution of a PB random variable is called a PB distribution.

A PB distribution is uni-modal and its mode is attain at the mean up to rounding to a nearby integer (see \cite{samuels65}). It follows that if $X$ is a PB random variable, then the total variation distance between $X$ and $X+1$ is the value of $X$ at it's mode. We will use the following conclusion.
\begin{lem}\label{lemma: X X+1 distance}
Let $X$ be a PB random variable with $\mu=E[X]$. We have,
\[
d_{TV}(X,X+1)=\max_{t\in\mathbb Z}P(X=t)=\max\{P(X=\lfloor\mu\rfloor),P(X=\lceil\mu\rceil)\}.
\]
\end{lem}

A PB  distribution with a large variance can be approximated by a normal distribution with the same mean and variance in a very strong sense. Let $\phi(x)=\frac 1 {\sqrt{2\pi}}e^{\frac{x^2}{2}}$ be the the standard normal density. The following lemma is taken from Pitman \cite[Eq. (25)]{pitman97} who attributes it to Platnov~\cite{platonov80}.
\begin{lem}\label{lemma: normal approximation}
Let $X$ be a PB random variable with $\mu=E[X]$, and $\sigma^2=Var[X]$. For every $t\in\mathbb Z$,
\[
\left\vert \sigma P(X=t)-\phi\left(\frac{t-\mu}{\sigma}\right)\right\vert \leq \frac C \sigma,
\]
for some global constant $C$.
\end{lem}

\section{Proofs}

We denote the indicator vector of an action $j\in[k]$ by $e_j\in\mathbb R^k$. For a strategy profile $a=(a_1,\ldots,a_n)\in[k]^n$,  define
\[
N(a)=\sum_{i=1}^n e_{a_i}.
\]
Namely, $N(a)\in\mathbb Z_+^k$ is the vector that counts the number players who take each one of the actions. Since a perturbed action  profile $a^\delta=(a_1^\delta,\ldots,a_n^\delta)$ is a random variable that takes values in $[k]^n$, $N(a^\delta)$ is a random variable that takes values in $\mathbb Z_+^k$. Given an anonymous game $g\colon[k]^n\to[0,1]^n$, and a player $i\in[n]$, $g_i(\cdot)$ is a function of $a_i$ and $N(a_{-i})$; therefore, for any action $a_i\in[k]$, $g_i^\delta(a_i,a_{-i})=E[f(N(a_{-i}^\delta))]$, for some function $f\colon \mathbb Z_+^k\to[0,1]$. Since any such $f$ can be realised by setting $g_i(a_i,a_{-i})=f(N(a_{-i}))$, 
\[
\lambda(n,k,\delta)=\max_{f,a,b}E\left[f(N(a^\delta))\right] - E\left[f(N(b^\delta))\right],
\]
where the maximum if over all $f\colon\mathbb Z_{+}^k\to[0,1]$, and $a,b\in[k]^{n-1}$ subject to $\rho(a,b)=1$. The maximum on the right-hand side is attained when $f$ achieves the total variation distance between $N(a^\delta)$ and $N(b^\delta)$; therefore, by (arbitrarily) fixing the place
in which $a$ and $b$ differ, we have
\begin{multline}\label{equation: total variation}
\lambda(n,k,\delta)=\max_{a\in[k]^{n-2}} d_{TV}\left(e_1^\delta+N(a^\delta), e_2^\delta+N(a^\delta)\right)\\
=  (1-\delta)\max_{a\in[k]^{n-2}} d_{TV}\left(e_1+N(a^\delta), e_2+N(a^\delta)\right),
\end{multline}
where $d_{TV}(\cdot,\cdot)$ denotes the total variation distance.

\subsection{Proof of Theorem~\ref{theorem: k=3}}

In light of \eqref{equation: total variation}, the next lemma implies the upper bound of Theorem~\ref{theorem: k=3}.
\begin{lem}\label{lemma: reflection}
For every $k\geq 2$, $n\geq 1$,  and $0<\delta<1$, 
\[
\max_{a\in[k]^n} d_{TV}(e_1+N(a^\delta),e_2+N(a^\delta))\leq P(S_n^{2\delta/k}\in\{0,1\}).
\]
\end{lem}
\begin{proof}
Let $a\in[k]^n$ be arbitrary. Let $X_1,\ldots,X_n\in\{e_1,\ldots, e_k\}$ be independent random vectors indicating the realisations of $a_1^\delta,\ldots ,a_n^\delta$ respectively. Namely, 
\[
P(X_i=e_j)=\begin{cases}
1-\delta + \frac \delta k & j = a_i,\\
\frac \delta k & j \neq a_i,
\end{cases}
\]

We would like to construct a coupling $(Z_n,Z'_n)$ such that $Z_n\sim e_1+\sum_{i=1}^n X_i$, $Z'_n\sim e_2+\sum_{i=1}^n X_i$ and $P(Z_n\neq Z'_n)\leq \eta^{-\frac 1 2} + \mathcal O(\eta^{-1})$. To this end, we define random variables $X'_1,\ldots,X'_n$ that have the same joint distribution as $X_1,\ldots,X_n$, and let $Z_m=e_1+\sum_{i=1}^m X_i$ and $Z'_m=e_2+\sum_{i=1}^m X'_i$, for every $m=1,\ldots,n$.

Informally, each $X'_i$ is going to be defined to be $e_2$ wherever $X_i=e_1$ and vice-versa, as long as $Z_{i-1}\neq Z'_{i-1}$, and $X'_i=X_i$ otherwise (either if $X_i\notin\{e_1,e_2\}$, or once $Z_{i-1}=Z'_{i-1}$). 

Formally, the random variables $X_1,\ldots,X_n$ are realised as follows:
\[
X_i= \chi_i e_{U_i} + (1-\chi_i)e_{a_i},
\]
where $\chi_1,\ldots,\chi_n \sim \mathrm{Bernoulli}(\delta)$, $U_1,\ldots,U_n\sim \mathrm{Uniform}([k])$ are all independent random variables.

The $X'_i$-s are coupled with the $X_i$-s through the following definition:
\[
X'_i= \chi_i e_{U'_i} + (1-\chi_i)e_{a_i},
\]
where $U'_1,\ldots,U'_n$ are defined recursively by
\[
U'_i = \begin{cases}
3-U_i& Z_{i-1}\neq Z'_{i-1}\text{ and } U_i\in[2],\\
U_i & \text{otherwise,}
\end{cases}
\]
setting $Z_0=e_1$, $Z'_0=e_2$. 

We explain why $X'_1,\ldots,X'_n$ are indeed independent random variables with $X_i\sim a_i^\delta$ for every $i\in [n]$. Let $\mathcal F_i=\sigma\langle \chi_1,U_1,\ldots,\chi_{i},U_{i}\rangle$. By its definition, the distribution of $U'_i$ is uniform in $[k]$ conditioned on $\mathcal F_{i-1}$, for every $i$; therefore $X'_i\sim a^\delta_i$ conditioned on $\mathcal F_{i-1}$. Furthermore, $X'_i$ is $\mathcal F_i$-measurable ; therefore $X'_i\sim a^\delta_i$ conditioned on $X'_1,\ldots,X'_{i-1}$.

The definition of $X'_i$ is so that  $Z_i=Z'_i$ implies that $Z_{i+1}=Z'_{i+1}$, for every $i\in[n-1]$; therefore $Z_n=Z'_n$ iff there exists $i\in[n]$ such that $Z_i=Z'_i$. Furthermore, for every $0\leq i\leq n$ and $3\leq j\leq k$, $(Z_i)_j=(Z'_i)_j$ and $(Z_i)_1+(Z_i)_2=(Z'_i)_1+(Z'_i)_2$; therefore  $Z_n=Z'_n$ iff there exists $1 \leq i\leq n$ such that $(Z_i)_1=(Z'_i)_1$. 

Let $S_i := 1- (Z_i)_1+(Z'_i)_1$, $i=0,\ldots, n$. Note that $S_i$ is almost a symmetric random walk on $\mathbb Z$ with the exception that it stays put forever once it hits 1. A direct calculation  shows that conditioned on $S_{i}\neq 1$,
\begin{equation}
\label{equation:random walk}
S_{i+1}=
\begin{cases}
S_i   & \text{w.p.}\quad 1-\frac{2\delta}{k},\\
S_i + 1 & \text{w.p. }\quad \frac{\delta}{k},\\
S_i - 1 & \text{w.p. }\quad \frac{\delta}{k}.
\end{cases}
\end{equation}

Since \eqref{equation:random walk} is exactly the rule of $(S_i^{2\delta/k})_{i=0}^\infty$ (unlike $S_i$, $S_i^{2\delta/k}$ does not stop when it hits 1), Lemma~\ref{lemma: reflection principle} concludes the proof of Lemma~\ref{lemma: reflection}. 
\end{proof}

The following lemma states that the upper-bound of Lemma~\ref{lemma: reflection} is tight in case that $k\geq 3$.
\begin{lem}\label{lemma: lower bound}
For every $k\geq 3$, $n\geq 1$,  and $0<\delta<1$,
\[
\max_{a\in [k]^n}d_{TV}(e_1+N(a^\delta),e_2+N(a^\delta))\geq P(S_n^{2\delta/k}\in\{0,1\}). 
\]
\end{lem}
\begin{proof}
Consider the strategy profile $\mathtt{\bar 3}\in[k]^n$ in which all of the players take action 3. Let $X$ be the random variable that counts the difference between the number of players who play 1 and those who play 2 under the mixed strategy profile $\mathtt{\bar 3}^\delta$. Formally, define $f\colon \mathbb Z^k\to \mathbb Z$ by $f(x_1,\ldots,x_n)=x_1-x_2$. Then, $X:=f(N(\mathtt{\bar 3}^\delta))$. Since $f(e_1+N(\mathtt{\bar 3}^\delta))=X+1$ and $f(e_2+N(\mathtt{\bar 3}^\delta))=X-1$,
\begin{multline*}
d_{TV}(e_1+N(\mathtt{\bar 3}^\delta),e_2+N(\mathtt{\bar 3}^\delta))\geq d_{TV}(X+1,X-1) \\
\geq P(X+1>0)-P(X-1>0)=P(X\in\{0,1\}).
\end{multline*}

The proof of Lemma~\ref{lemma: lower bound} is concluded since $X\sim S_n^{2\delta/k}$.
\end{proof}

\subsection{Proof of Theorem \ref{theorem: k=2}}
Let $X_1,X_2,\ldots$ be i.i.d.\ $\mathrm{Bernoulli}(\delta/2)$ random variables. Define
\[
M(n,\delta)=\max_{l,s\in\{0,\ldots,n\}}P\left(\sum_{i=0}^l X_i + \sum_{j=l+1}^n (1-X_j)=s\right),
\]
and $M(0,\delta)=1$ by convention.
\begin{lem}
\label{lemma: k=2}
For every $n\geq 2$ and $\delta\in(0,1)$,
\[
\lambda(n,2,\delta)=(1-\delta)M(n-2,\delta).
\]
\end{lem}
\begin{proof}
By \eqref{equation: total variation} it is sufficient to prove that 
\begin{equation*}
\max_{a\in[k]^{n}} d_{TV}\left(e_1+N(a^\delta), e_2+N(a^\delta)\right) = M(n,\delta). 
\end{equation*}
For every $a\in[2]^n$ there is an $l\in \{0,\ldots,n\}$ such that $N(a)=(l,n-l)$ and vice versa; therefore it is sufficient to prove that for each such pair $a\in[2]^n$ and $l\in\{0,\ldots,n\}$,
\begin{multline*}
d_{TV}\left(e_1+N(a^\delta), e_2+N(a^\delta)\right) \\
= \max_{s\in\{0,\ldots,n\}}P\left(\sum_{i=0}^l X_i + \sum_{j=l+1}^n (1-X_j)=s\right). 
\end{multline*}

Let $X$ be the random variable that counts the number of players who play 1 under the mixed strategy profile $a^\delta$. Formally, $X$ is defined by $N(a^\delta)= (X,n-X)$. Let $f\colon x\mapsto n+1-x$. Since $e_1+N(a^\delta)=(X+1,f(X+1))$ and $e_2+N(a^\delta)=(X,f(X))$,
\[
d_{TV}(e_1+N(a^\delta),e_2+N(a^\delta))=d_{TV}(X+1,X).
\]
Since $X$ is PB, by Lemma~\ref{lemma: X X+1 distance}, 
\[
d_{TV}(X+1,X)=\max_s P(X=s).
\]
The proof is concluded since $X\sim \sum_{i=0}^l X_i + \sum_{j=l+1}^n (1-X_j)$.
\end{proof}
The next lemma states that the maximisers in the definition of $M(n,\delta)$ are $s=l=n/2$, for $n$ even, and it provides upper- and lower-bounds, for $n$ odd.
\begin{lem}\label{lemma: esg}
For every $n\in\mathbb N$ and $\delta\in(0,1)$, let 
\[P_n=P\left(\sum_{i=1}^n X_i+\sum_{j=n+1}^{2n}(1-X_j)=n\right).\]
Then,
\[
P_{\lceil n/2 \rceil}\leq M(n,\delta) \leq \sqrt{P_{\lceil n/2 \rceil}P_{\lfloor n/2 \rfloor}}
\]
\end{lem}
\begin{proof}
To prove the first inequality, $P_{\lceil n/2 \rceil}\leq M(n,\delta)$, it is sufficient to show that $M(n,\delta)$ is decreasing in $n$, and $P_n\leq M(2n,\delta)$, for every $n\in\mathbb N$. The latter follows  from the definition of $M(n,\delta)$  directly. The former holds since, there are some $l,s_0,s_1$, such that 
\begin{multline*}
M(n+1,\delta)=P(X_{n+1}=0)P\left(\sum_{i=1}^{l} X_i+\sum_{j=l+1}^{n}(1-X_j)=s_0\right)\\
+P(X_{n+1}=1)P\left(\sum_{i=1}^{l} X_i+\sum_{j=l+1}^{n}(1-X_j)=s_1\right)\leq M(n,\delta).
\end{multline*}

It remains to prove the second inequality $M(n,\delta)\leq \sqrt{P_{\lceil n/2 \rceil}P_{\lfloor n/2 \rfloor}}$. Let $l$ and $s$ be such that $M(n,\delta)=P\left(\sum_{i=1}^{l} X_i+\sum_{j=l+1}^{n}(1-X_j)=s\right)$. Define $\epsilon_1,\ldots,\epsilon_n\in\{+1,-1\}$ by $\epsilon_i=+1$ ($i\leq l$) and $\epsilon_i=-1$ ($i>l$). Let $Y=\sum_{i=1}^{\lceil n/2 \rceil}\epsilon_i X_i$ and $Z=\sum_{i=\lceil n/2 \rceil+1}^{n}\epsilon_i X_i$. Then, by Cauchy-Schwarz Inequality,
\begin{multline*}
M(n,\delta) =P\left(Y+Z=s-n\right)=\sum_{t}P(Y=y)P(Z=s-n-t)\\
\leq \sqrt{\sum_{t}\left(P(Y=t)\right)^2\sum_{t}\left(P(Z=t)\right)^2}.
\end{multline*}
The proof will be concluded by showing that $\sum_{t}\left(P(Y=t)\right)^2=P_{\lceil n/2 \rceil}$ and $\sum_{t}\left(P(Z=t)\right)^2=P_{\lfloor n/2 \rfloor}$. More generally, we show that for every $n$ and every $\epsilon_1,\ldots,\epsilon_n\in\{+1,-1\}$, letting $X=\sum_{i=1}^n \epsilon_i X_i$,
\begin{equation}\label{equation: P_n}
\sum_{t}\left(P\left(X=t\right)\right)^2=P_n.
\end{equation}
Since
\[
\sum_{t}\left(P\left(X=t\right)\right)^2 = P(X=X'),
\]
where $X'$ is an independent copy of $X$, the case of $\epsilon_1=\cdots =\epsilon_n$ is evident. It remains to show that toggling one of the $\epsilon_i$-s does change the quantity at the left-hand side of \eqref{equation: P_n}. More generally, we show that for any two independent discrete random variables, $X$ and $Y$
\[
P(X+Y = X'+Y') = P(X-Y=X'-Y'),
\]
where $X',Y'$ are independent copies of $X,Y$. This is true since, 
\[
P(X+Y = X'+Y') = P(X-Y'=X'-Y)= P(X-Y=X'-Y').
\]
\end{proof}
The proof of Theorem~\ref{theorem: k=2} follows immediately from Lemmata~\ref{lemma: k=2} and \ref{lemma: esg}, since  $S_n^{\delta(1-\delta/2)}\sim \sum_{i=1}^n \left(X_i-X_{n+i}\right)$.
\subsection{Proof of Theorem~\ref{theorem:main}}
By Theorem~\ref{theorem: k=2}, the second part of Theorem~\ref{theorem:main}, the case $k=2$, follows from the next claim.
\begin{claim}\label{claim: asymptotic r<1/2}
For every $n\in\mathbb N$, and $0<r\leq \frac 1 2$, let $\eta = 2\pi n r$. Then,
\[
\left\vert P\left(S_{n}^{r}=t\right)-\eta^{-\frac 1 2}\right\vert = \mathcal O(\eta^{-1}),\quad\quad t=0,1.
\]
\end{claim}
\begin{proof}
Since $r \leq \frac 1 2$, the increments $S^r_{i+1}-S^r_i$ can be realised as the difference of two i.i.d.\ Bernoulli random variables; therefore $S^r_n$ is a Poisson Binomial random variable. Applying Lemma~\ref{lemma: normal approximation} with $\mu = E[S^r_n]=0$, $\sigma^2= Var[S^r_n]=nr$ gives
\[
\left\vert \sigma P(S^r_n=t)- \frac {1}{\sqrt{2\pi}}\right\vert  = \mathcal O(\sigma^{-1}),\quad\quad t=0,1,
\]
which concludes the proof of Lemma~\ref{claim: asymptotic r<1/2}.
\end{proof}

For the first part of  Theorem~\ref{theorem:main}, the case $k\geq 3$, we have to consider $r>\frac 1 2$ as well. 
\begin{claim}\label{claim: asymptotic r>1/2}
For every $n\in\mathbb N$ and $0<r\leq 1$, let $\eta = \frac 1 2 \pi nr$. Then, 
\[
\left\vert P(S^r_n\in\{0,1\})- \eta^{-\frac 1 2}\right\vert  = \mathcal O(\eta^{-1}).
\]
\end{claim}
\begin{proof}
The case $r\leq \frac 1 2$ is an immediate consequence of Claim~\ref{claim: asymptotic r<1/2}. The case $r=1$ and $n$ even, too,  follows from Claim~\ref{claim: asymptotic r<1/2}, since $S^1_{2m}\sim 2 S_m^{\frac 1 2}$, and $P(S^1_{2m}=1)=0$. By Lemma~\ref{lemma: reflection principle}, $P(S_n^1\in\{0,1\})$ is monotonic in $n$; therefore the claim holds for $r=1$ and $n$ odd, as well.

For $r\in\left(\frac 1 2,1\right)$, let us realise $S^r_n$ as $S^1_{X}$, where $X\sim\mathrm{Binomial}(n,r)$ independently of $(S^1_i)_{i=1}^n$. Let $C_1$ be the constant of the ``$\mathcal O$'' term in the claim for the case $r=1$. Let $f(x,r)=(1+x/r)^{-\frac 1 2}$, and $C_2=\max \{|\mathrm d f/\mathrm d x|:\frac 1 2 \leq r\leq 1 ,\ |x|\leq \frac 1 4\}$. The proof of Claim~\ref{claim: asymptotic r>1/2}, and  Theorem~\ref{theorem:main} thereby, is concluded as follows:
\begin{multline*}
    \left\vert P(S^r_n\in\{0,1\})- \eta^{-\frac 1 2}\right\vert = \left\vert P(S^1_X\in\{0,1\})- \eta^{-\frac 1 2}\right\vert \\
    \leq P(|X-nr|>n/4)\ +\\
    \sum_{k:|k-nr|\leq n/4}P(X=k)\left(\left\vert P(S_k^1\in\{0,1\})-\left(\tfrac 1 2 \pi k\right)^{-\frac 1 2}\right\vert+\left\vert\left(\tfrac 1 2 \pi k\right)^{-\frac 1 2}-\left(\tfrac 1 2 \pi nr\right)^{-\frac 1 2}\right\vert\right)
    \\ 
    \leq \frac{16 Var[X]}{n^2}+ \sum_{k:|k-nr|\leq n/4}P(X=k)\left( C_1\frac 2 {\pi k}
    +\eta^{-\frac 1 2}\left\vert f(k/n-r,r)-f(0,r)\right\vert\right)\\
    \leq \frac{4}{n} +  \frac{C_1 8}{\pi n}+ \eta^{-\frac 1 2}C_2\sum_{k:|k-nr|\leq n/4}P(X=k)\left\vert k/n-r\right\vert\\
    \leq  \mathcal O (n^{-1}) + C_2\eta^{-\frac 1 2} E\left\vert \tfrac X n -r\right\vert\\
    \leq \mathcal O (n^{-1}) +C_2\eta^{-\frac 1 2}\sqrt{Var\left[\tfrac X n\right]} = \mathcal O (n^{-1}).
\end{multline*}
The second and the last inequalities use Chebyshev's and Jensen's Inequalities respectively, and $Var[X]=nr(1-r)\leq n/4$.
\end{proof}

\subsection{Proof of Theorem \ref{theorem: approximate nash}}
Let $g$ be an $n$-player $k$-action anonymous game. Every $\epsilon$-Nash equilibrium in $g^\delta$ is a $(\delta+\epsilon)$-Nash in $g$; therefore the first part of Theorem~\ref{theorem: approximate nash} is an immediate consequence of the following theorem.
\begin{thm}[Theorem 6.1 in ~\cite{azrieli-shmaya13}]
Any $n$-player $k$-action $\lambda$-Lipschitz game admits a $2k\lambda$-Nash equilibrium in pure strategies.
\end{thm}
The second part of Theorem~\ref{theorem: approximate nash} follows from Theorem~\ref{theorem:main}, and setting $\delta=\lambda(n,k,\delta)$, and $\epsilon=2\delta$.
\section{Acknowledgement}
We are grateful to \href{https://mathoverflow.net/}{MathOverflow} for facilitating this collaboration.
We thank the members of the Facebook group \url{https://www.facebook.com/groups/305092099620459/} for pointing out references related to the reflection principle. 
\bibliographystyle{plain}
\bibliography{mybib}
\end{document}